\documentclass[10pt, conference]{IEEEtran}

\usepackage{graphicx}
\usepackage{amssymb,amsmath,graphicx}
\usepackage{algorithm}
\usepackage{algorithmic}

\newtheorem{theorem}{Theorem}

\newtheorem{lemma}{Lemma}

\newtheorem{corollary}{Corollary}

\begin{document}

% paper title
\title{Improved Delay Estimates for a Queueing Model for Random Linear Coding for Unicast}

\author{Mohammad Ravanbakhsh$^1$, \'{A}ngela I. Barbero$^2$, and \O yvind Ytrehus$^1$\\
$^1$Dept. of Informatics, University of Bergen, N-5020 Bergen, Norway,\\
$^2$Dept. of Applied Mathematics, University of Valladolid, 47011 Valladolid, Spain,\\ Email: $^1$\{mohammad.ravanbakhsh,oyvind\}@ii.uib.no, $^2$angbar@wmatem.eis.uva.es}
\maketitle

\begin{abstract}
Consider a lossy communication channel for unicast with zero-delay feedback. For this communication scenario, a simple retransmission scheme is optimum with respect to delay.
An alternative approach is to use random linear coding in automatic repeat-request (ARQ) mode. We extend the work of Shrader and Ephremides in \cite{Shrader-MILCOM07}, by deriving an expression for the delay of random linear coding over a field of infinite size. Simulation results for various field sizes are also provided.
\end{abstract}


\section{Introduction}
Consider a communication model where packets arrive from an information source to a network node, the \emph{sender}, according to a deterministic or random arrival pattern, and are transmitted to another network node, the \emph{receiver}, according to a deterministic or random service pattern. Clearly, in a balanced system, the average rate at which packets arrive cannot exceed the average rate at which packets can be transmitted. However, even when the arrival rate is lower than the effective transmission rate, queues can build up due to randomness in the processes involved, thus increasing the delay of packet delivery.

Queuing theory provides a natural framework for the study of the delay of protocols operating on such communication models. Shrader and Ephremides \cite{Shrader-MILCOM07} compared the delay experienced by random linear packet-based coding to that of a basic retransmission scheme, for a simple \emph{lossy unicast} model.
\subsection{Contributions of this work}
In this paper we refine the delay estimates produced in \cite{Shrader-MILCOM07}. In particular, for the case of coding over an infinite (or very large) field, in Section~\ref{subsec:codInf} we derive an exact expression for the delay of a random linear packet-based coding ARQ scheme.
Section~\ref{subsec:Finite} presents improved estimates of the delay for the finite field case, and in Section~\ref{subsec:Sim} we have collected some results from computer simulations that we have performed.
\subsection{Related work}
As far as we know, \cite{Shrader-MILCOM07} is the only paper in the literature to deal with this specific model.  Related but  different work in a unicast setting can be found in \cite{ErOzMe06,ErOzMeAh08}, but in the model addressed in these papers packets arrive to the sender as a block, so there is no loss in waiting. Reference \cite{LuPaFraMeKo06} addresses delay in forward-error correction schemes for sender nodes with a finite buffer capacity.

Random linear coding \cite{Luby} is known to work particularly well in the multicast case \cite{HoKoMeKa03}, and multicast problems similar to the ones we discuss here are investigated in \cite{Shrader-ITW06,SuShMe08}.

\subsection{Outline} This paper is structured as follows: Section~\ref{sec:bg} contains background material, including a description of the queueing model, basic notation, and a brief overview over prior work. In Section~\ref{sec:nr} we present new results, including the determination of the delay for coding over an infinite field and bounds for the finite field case, and results obtained by simulation. We present conclusions in Section~\ref{sec:sum}.
\section{Background}
\label{sec:bg}
Our departure point in this paper is Shrader and Ephremides' work in \cite{Shrader-MILCOM07}. In this section we summarize necessary background material, including a description of the communication model, basic notation, and some prior results.
\subsection{Notation and model description}
The communication model considered in  \cite{Shrader-MILCOM07} and in this paper contains the  elements listed below.

\begin{enumerate}
 \item A \emph{sender} with unlimited buffer memory.
 \item A \emph{packet source} that injects \emph{fixed-length} packets into the sender. \emph{In this paper,} as in \cite{Shrader-MILCOM07}, time is \emph{slotted}, and the packets arrive at the sender through a Bernoulli process with \emph{arrival rate} $\lambda$. Hence, during a given time slot, the number of new packets injected into the sender is either zero or one, with probabilities $1-\lambda$ and $\lambda$, respectively.
 \item A \emph{receiver} with unlimited patience.
 \item A \emph{channel} for sending packets from the sender to the receiver, and a \emph{feedback channel} for sending acknowledgements (or negative acknowledgements) from the receiver back to the sender. The channels have the following properties:
\begin{itemize}
 \item During a time slot, the sender may transmit exactly one fixed-length packet.
 \item A packet transmitted by the sender is successfully delivered to the receiver with probability  $q$, $ 0 <q \leq 1$. Two distinct packet transmissions are independent. Thus the average service time of the sender is $1/q$ for a single packet.
 \item There is \emph{no delay} between sending a packet and receiving a positive or negative acknowledgement. I. e., if a packet has been transmitted during a given time slot, it will be known at the start of the next time slot whether the transmission was successful or not.
\end{itemize}
 \item A \emph{communication protocol} that regulates the flow of packets between the sender and the receiver. Due to the assumption of zero acknowledgement delay, an automatic repeat-request (ARQ) scheme is optimum with respect to number of packets sent. The two schemes compared in \cite{Shrader-MILCOM07} are
\begin{itemize}
 \item \emph{Retransmission} of single packets: The sender collects arriving packets in a queue. When the queue is non-empty, the server selects a packet from the queue (without loss of generality, the queue can be ordered on a first in, first out basis), transmits the packet in the first available time slot, and continues to send that packet until a successful delivery is confirmed. The precise retransmission scheme (i.e stop-and-go, go-back-N, selective repeat) is not important in this context since there is no acknowledgement delay.
 \item \emph{Random linear coding:} When the queue is non-empty, the sender selects the first $k$ packets, $I_1,\ldots,I_k$,  in the queue (where $k = \min \{$ number of packets in queue, $K\}$ and $K$ is some  pre-selected integer that characterizes the protocol). The collection $I_1,\ldots,I_k$ will be referred to as a \emph{bulk}. The sender proceeds to generate encoded packets based on the bulk. (A bulk service approach for forward error correction has been studied in \cite{SaMu06}.) Each encoded packet $\tilde{I}$ is a random (or pseudo-random) linear combination \[\tilde{I}=\sum_{i=1}^k \alpha_i(\tilde{I}) I_i\] of the original information packets, where $\alpha_i(\tilde{I})$ are coefficients selected at random for each encoded packet $\tilde{I}$ from some field $\mathbb{F}$. We assume that the coefficients $\alpha_i(\tilde{I})$ can be conveyed to the receiver at no significant cost. The sender proceeds to generate and send randomly encoded packets until the receiver has successfully collected a set $\{\tilde{I}_1,\ldots,\tilde{I}_{k'}\}$, where $k' \geq k$, of encoded packets from which all of the original information packets $I_1,\ldots,I_k$  in the bulk can be recovered. Upon successful delivery, the sender returns to inspect its packet queue and will serve another bulk of packets. The protocol will be referred to as RLC($K,\mathbb{F}$), and is consistent with e. g. fountain coding \cite{Luby} or other random linear coding schemes. Note that the special case of $K=1$ corresponds to  retransmission  (for the only sensible coding scheme.)
\end{itemize}

\end{enumerate}

Figure~\ref{fig:systemmodel} summarizes the model.

\begin{figure}[htbp]
\vspace{-1.0cm}
%\centerline{\includegraphics[width=12cm,keepaspectratio=true,angle=0]{figures/systemmodel.pdf}}
\centerline{\includegraphics[width=9cm,keepaspectratio=true,angle=0]{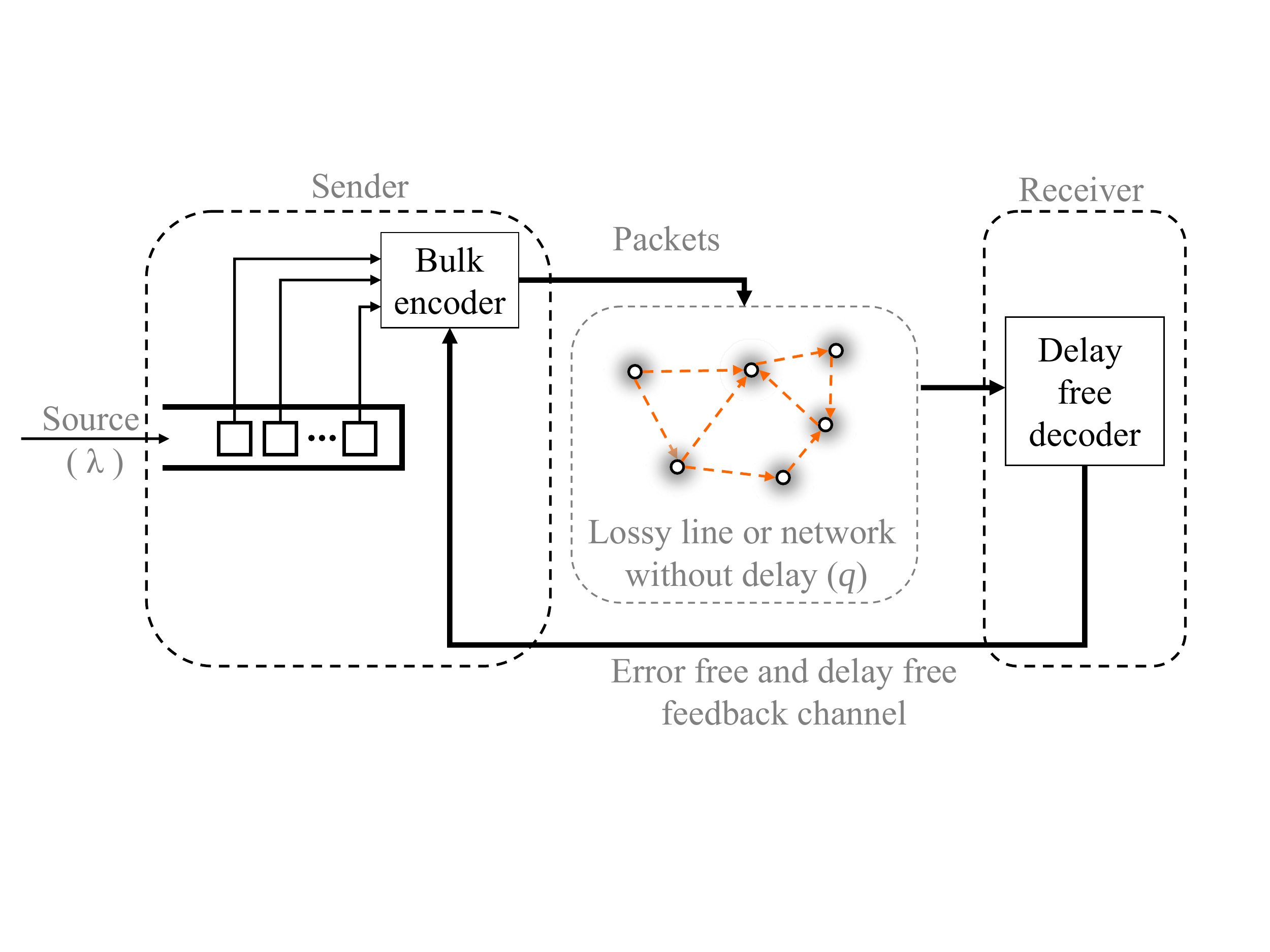}}
\vspace{-1.5cm}
\caption{System Model.}\label{fig:systemmodel}
\end{figure}

In a practical general communication model, there are many measures of performance that conceivably may be of interest, such as throughput, sender's delay, receiver's delay, reliability, efficiency in terms of usage of resources, and processing complexity. In this paper we focus on the \emph{sender's delay}, under the implicit assumptions that the communication protocol ensures a completely reliable packet delivery, that processing cost or resource usage are of no importance, that the receiver's delay is precisely determined by the sender's delay, and that the throughput is given exactly by the average packet arrival rate and limited by the channel's service rate. By the average sender's delay we mean the average number of time instances between the moment when the packet is injected into the sender's queue and the moment when the packet has been successfully delivered.
\subsubsection*{Extension to the model in \cite{Shrader-MILCOM07}}
Basically we consider the same model and the same measures of performance as those in \cite{Shrader-MILCOM07}, with the following minor extensions:
\begin{itemize}
 \item In \cite{Shrader-MILCOM07}, binary coding is considered. Since coding is more efficient over a larger field, we will consider also larger fields, and in the limit, coding over an infinitely large field.
 \item For convenience, the authors of \cite{Shrader-MILCOM07} assume an encoding scheme that allows (except for a bulk size of 1) the choice of an all-zero encoding vector (i.e. all coefficients are zero.) This simplifies the analysis, but it is not ideal, especially for short bulk lengths and small fields, and it is easy to avoid in practice.
% \item The average delay is the most important parameter, but we note that delay \emph{variance} %in many practical situations may
\end{itemize}

%\subsection{Notation}
\subsection{Prior work}
\subsubsection{Basic retransmission} The retransmission protocol has been analyzed in \cite{Taka}. The average delay $D_{RE}$ and the average waiting time $W_{RE}$ are given by
\begin{equation}
\label{eq:delRE}
D_{RE} = \frac{1}{q} +W_{RE}, \mbox{ and } W_{RE} = \frac{\lambda(1-q)}{q(q-\lambda)}.
\end{equation}

\subsubsection{Random linear coding}
Shrader and Ephremides \cite{Shrader-MILCOM07} analyzed the RLC($K,\mathbb{F}_2$) protocol, where  $\mathbb{F}_2$ is the binary field. They observed that the protocol corresponds to a variant of what is known in queuing theory as a bulk service model \cite{Chaund,GroShoThoHar08}: a bulk of $k (\geq 1)$ packets are serviced simultaneously. Consequently they apply techniques for bulk service queues in order to obtain estimates of the delay.

Consider a bulk of size $k$, i. e. $k$ packets, $I_1,\ldots,I_k$ are selected for simultaneous service by the random coding protocol. The service time $X_k$ for the bulk is a random variable that denotes the number of packets that need to be sent (i. e., the number of time slots used) until the bulk has been successfully delivered. Let $b_k(z) = \sum_{i=0}^{\infty}b_{k,i}z^i$ be the probability generating function (p.g.f) for $X_k$, i. e. \[b_{k,i} = \mbox{Pr}\{\mbox{a size }k\mbox{-bulk is serviced in }i\mbox{ time slots }\}.\]  Note that in general,  $b_k(z)$ depends on $q$, $\mathbb{F}$, and  on the way in which the encoding coefficients are selected. For convenience, $b_0(z) = b_1(z)$ in order to deal with the situation when the queue is empty.
\begin{lemma}
Let $S_t$ be the number of packets in the queue immediately after the $t$-th bulk has been successfully delivered (and before service is initiated for the ($t+1$)-th bulk.) Let
$P_{k} = \lim_{t \rightarrow \infty}$ Pr$\{ S_t = k\}$ for $k=0,1,2,\ldots$,
and let $P_{K,\mathbb{F}}(z) = \sum_{i=0}^{\infty}P_iz^i$.
Then, as explained in  \cite{Shrader-MILCOM07},
\begin{equation}
\label{eq:P(z)}
 P_{K,\mathbb{F}}(z) = \frac{\displaystyle\sum_{k=0}^{K-1}P_k(z^Kb_k(\lambda z+1-\lambda)-z^kb_K(\lambda z+1-\lambda))}
{z^K-b_K(\lambda z+1-\lambda)}.
\end{equation}

\end{lemma}
\vspace{1cm}
Observe that $P_{K,\mathbb{F}}(z)$ is completely determined by $P_0,\ldots,P_{K-1}$. The authors of \cite{Shrader-MILCOM07} proceed by
\begin{itemize}
 \item applying techniques involving Rouch\'{e}'s theorem to obtain numerical values for $P_0,\ldots,P_{K-1}$,
 \item finding the expected number of packets in the queue \emph{immediately after a bulk has been successfully delivered}  to be
\begin{equation}
\overline{S} = \sum_{i=0}^{\infty}i P_i = dP(z)/dz|_{z=1},
\end{equation}
 and finally
 \item using $\overline{S}$ as an approximation to the average number of packets in the system \emph{at an arbitrary point in time,} and applying this with Little's law \cite{GroShoThoHar08,Kleinrock75} to obtain an approximation to the average delay $D_{RLC}$,
\begin{equation}
\label{eq:DRLCApprox}
D_{RLC} \approx \overline{S}/ \lambda.
\end{equation}
\end{itemize}
Note that although equation (\ref{eq:DRLCApprox}) is derived in \cite{Shrader-MILCOM07} for the case of binary encoding, the derivation remains valid for encoding in other fields provided  $b_{k}(z)$ is adjusted accordingly.
\section{New results}
\label{sec:nr}
Simulations indicate that (\ref{eq:DRLCApprox}) is a rather crude (under)estimate for the delay suffered for random linear coding, and that the problem is that the p.g.f  $P(z)$ does not really resemble the corresponding p.g.f. $Q(z)=Q_{K,\mathbb{F}}(z) = \sum_{i=0}^{\infty}Q_iz^i$, where $Q_i$ is the probability of observing $i$ packets in the system (i. e. waiting in queue or being serviced) \emph{at an arbitrary point in time.} Intuitively, this can be explained in the following way: Suppose that a bulk of size $k$, where $k$ is ``large'', is being processed. This means that the sender is transmitting at full speed for, on average, $kE[X_k]$ time slots. The expected number of arriving packets that arrive during the processing of the bulk is $\lambda kE[X_k]$. By definition, in a balanced queue $\lambda < k/E[X_k] $, and unless $\lambda$ is very close to $q$, the queue length will have been substantially reduced by the completion of service of the bulk.

In this section we describe an alternative approach to determining  $D_{RLC}$.
\subsection{Coding in a field of infinite size}
\label{subsec:codInf} We start by considering the (in some sense) simplified case where coding takes place in a field $\mathbb{F}$ which is very large: Technically we investigate the case of a finite field of $|\mathbb{F}|$ elements in the limit with $|\mathbb{F}| \rightarrow \infty$. We denote this field by $\mathbb{F}_{\infty}$. The important condition is that when a bulk of size $k$ is being transmitted, the probability that the receiver can recover the information upon successful reception of $k$ encoded packets is very close to $1$: The theoretical results obtained in this section are close to those obtained by simulation with coding over
$\mathbb{F}_{16}$, for example.

Observe a random bulk being serviced by the random coding scheme, and let $B_k$ be the probability that the bulk is of length $k$, for $1 \leq k \leq K$.

\begin{lemma}
The bulk distribution is given by
 \begin{eqnarray}
  B_1 & = & P_0+P_1 \\
  B_k & = & P_k \mbox{ for } k=2,\ldots,K-1, \mbox{ and }\\
 B_K  & = & 1-\sum_{i=1}^{K-1}B_i,
 \end{eqnarray}
where the $P_k$ are the same as in  (\ref{eq:P(z)}).
\end{lemma}
\begin{proof}
Consider a moment when a new bulk is selected. If the queue is empty, the sender will wait until a packet arrives; hence $ B_1  =  P_0+P_1$. If the queue contains $i$ packets, $2 \leq i <K$ packets, the bulk will also consist of these $i$ packets; and if there are $K$ or more packets, the server will form a bulk of size $K$.
\end{proof}

Let $L = (L_1,L_2,\ldots)$ be a %binary
sequence corresponding to a specific pattern of packet arrivals at the sender, i.e. an instance of a source output sequence.
%Thus, $L_i=1$ if the source injects a packet into the sender at time slot $i$.
Also let $C = (C_1,C_2,\ldots)$ be another %binary
sequence representing a specific success pattern of the channel, i. e. an instance of a channel success sequence.
%, so that $C_i=1$ if the $i$-th attempt of sending a packet is successful.
\begin{lemma}
\label{lem:bsm}
Let $L$ and $C$ be two specific arrival and channel success patterns, respectively. Apply the two alternative protocols \emph{retransmission} and \emph{RLC($K,\mathbb{F_\infty}$)} simultaneously to separate copies of $L$ and $C$, and consider a moment just when the \emph{RLC($K,\mathbb{F_\infty}$)} is ready to inspect the queue to form a new bulk (either at the very start of time, or because a bulk has just been successfully delivered.)

At this moment, which for convenience we may call a \emph{bulk selection moment},
\begin{enumerate}
 \item the retransmission protocol is also ready to select a new packet for transmission, (either because this is at the very start of time, or because a packet has just been successfully delivered,)
 \item the set of all packets already delivered by the retransmission protocol is equivalent to the set of packets already delivered by \emph{RLC($K,\mathbb{F_\infty}$)}, and
 \item the set of packets waiting in the queue, and the time that each packet has spent waiting so far, is independent of which of the two protocols is used.
\end{enumerate}
\end{lemma}
\begin{proof}
  This can be seen by a simple induction argument. The claims are obviously true at the start of time. Next, suppose this is true at some bulk selection moment and that the \emph{RLC($K,\mathbb{F_\infty}$)} selects a bulk of size $k$. Then, since both protocols will complete service of the $k$ information packets precisely when the channel has successfully delivered $k$ packets, the claim will be true also at the next bulk service moment.
\end{proof}

Note that for time instances that are not bulk selection moments, claims 1-3 of Lemma~\ref{lem:bsm} in general do not hold.
\begin{theorem}
\begin{equation}
\label{eq:delRLC}
D_{RLC} =  W_{RE} + \frac{\sum_{k=1}^K B_k \cdot k \cdot (k+1)}{2q},
\end{equation}
\end{theorem}
\begin{proof}
Suppose that at a bulk selection moment, the \emph{RLC($K,\mathbb{F_\infty}$)} selects a bulk of size $k$. Because of Lemma~\ref{lem:bsm}, it makes sense to compare the processing of the bulk directly with the retransmission protocol. We proceed to derive how much each packet will be delayed when \emph{RLC($K,\mathbb{F_\infty}$)} is used, compared to the delay when retransmission is used. Assume that  the packets in the bulk are $I_1,\ldots,I_k$, where $I_1$ is the first packet in the queue. Assign a fictious waiting time for each packet in the bulk, corresponding to the waiting time that the packet would experience with a retransmission scheme, and define the modified service time for each packet to be the total delay minus the fictious waiting time. Then the modified service time of $I_j$ is  $(k-(j-1))E[X_1] = (k-j+1)/q$. Thus, for a bulk of size $k$, the average delay is \[D_{RLC} =  W_{RE} + \frac{1}{k}\frac{\sum_{i=1}^k  i}{q}= W_{RE} + \frac{k+1}{2q}.\]
The theorem follows by averaging the second term (i. e., the service time) over all bulks.
\end{proof}

\begin{corollary}
 The delay of random linear coding over an infinte field is at most a factor of
\begin{equation}
\label{eq:coRatio}
\frac { (\lambda -q) K (K+1)}{2q ( 1-\lambda) }+\frac {\lambda (1-q)}{q ( 1-\lambda ) }
\end{equation}

worse than the delay of retransmission. This factor approaches $1$ as $\lambda \rightarrow q$.
\end{corollary}
\begin{proof}
From (\ref{eq:delRE}) and (\ref{eq:delRLC}), the factor is
\[
\frac{D_{RLC}}{D_{RE}} =  \frac{W_{RE} + \frac{\sum_{k=1}^K B_k \cdot k \cdot (k+1)}{2q}}
{W_{RE} + \frac{1}{q}}.
\]
(\ref{eq:coRatio}) follows since $\sum_{k=1}^K B_k \cdot k \cdot (k+1) \leq K(K+1)$.
\end{proof}

\subsection{Coding over a finite field}
\label{subsec:Finite} A similar approach can be applied in the case of a finite field. In this case,  comparing the delay of the RLC scheme with that of retransmission does not give an exact expression for the delay, since the waiting time is affected by the extra packet retransmissions required due to decoding errors. A lower bound on the RLC delay can, however, be derived, by adding the service time of the RLC scheme (which can be computed from knowledge of the bulk distribution) to the waiting of the basic retransmission scheme. We omit the details.
\subsection{Simulation results}
\label{subsec:Sim}
We have simulated the queueing models for the protocols under consideration. The results are shown in the figures in this section. (We had to make sure that the samples from the pseudo-random generator appear uncorrelated, as the standard Java pseudo-random generator failed to meet this criterion.) Each simulation curve is based on 100 values of $\lambda$, and for each $\lambda$ the average delay is measured over a time scale of $10^7$ time slots. We do not show simulations for the infinite field case, as they coincide with the curves derived from (\ref{eq:delRLC}).

Figure~\ref{fig:sim1} presents results from \cite{Shrader-MILCOM07}, from Theorem 1, and from computer simulations, for channel success rate $q=0.5$, while Figure~\ref{fig:sim2}
shows similar curves for $q=0.9$.
\begin{figure}[htbp]
%\vspace{-1cm}
%\centerline{\includegraphics[width=12cm,keepaspectratio=true,angle=0]{figures/systemmodel.pdf}}
\centerline{\includegraphics[width=9cm,keepaspectratio=true,angle=0]{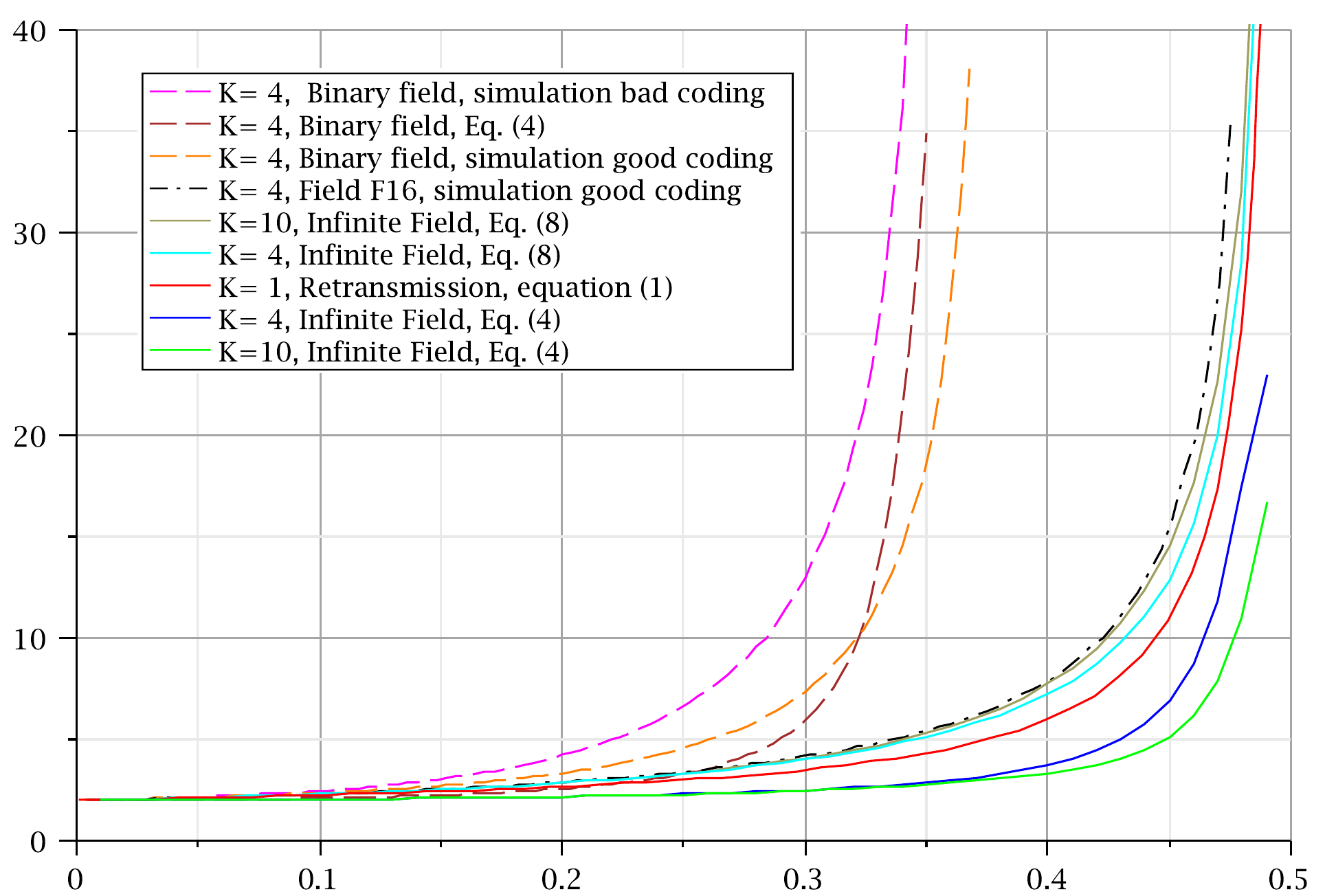}}
%\vspace{-2.0cm}
\caption{Average delay, in time slots, as a function of arrival rate $\lambda$ for $q=0.5$. Curves that refer to eq. \ref{eq:DRLCApprox} %in \cite{Shrader-MILCOM07}
are calculated by that equation, even though the curves as such may not appear in \cite{Shrader-MILCOM07}. The simulation curves marked ``bad coding'' is obtained by using the RLC coding scheme in \cite{Shrader-MILCOM07}, which allows  transmission of the zero codeword for bulk size larger than one.  The simulation curves marked ``good coding'' is obtained by using an RLC coding scheme that does not allow transmission of the zero codeword. The legend is ordered according to the saturation point of curves from left to right.}\label{fig:sim1}
\end{figure}

\begin{figure}[htbp]
%\vspace{-1cm}
%\centerline{\includegraphics[width=12cm,keepaspectratio=true,angle=0]{figures/systemmodel.pdf}}
\centerline{\includegraphics[width=9cm,keepaspectratio=true,angle=0]{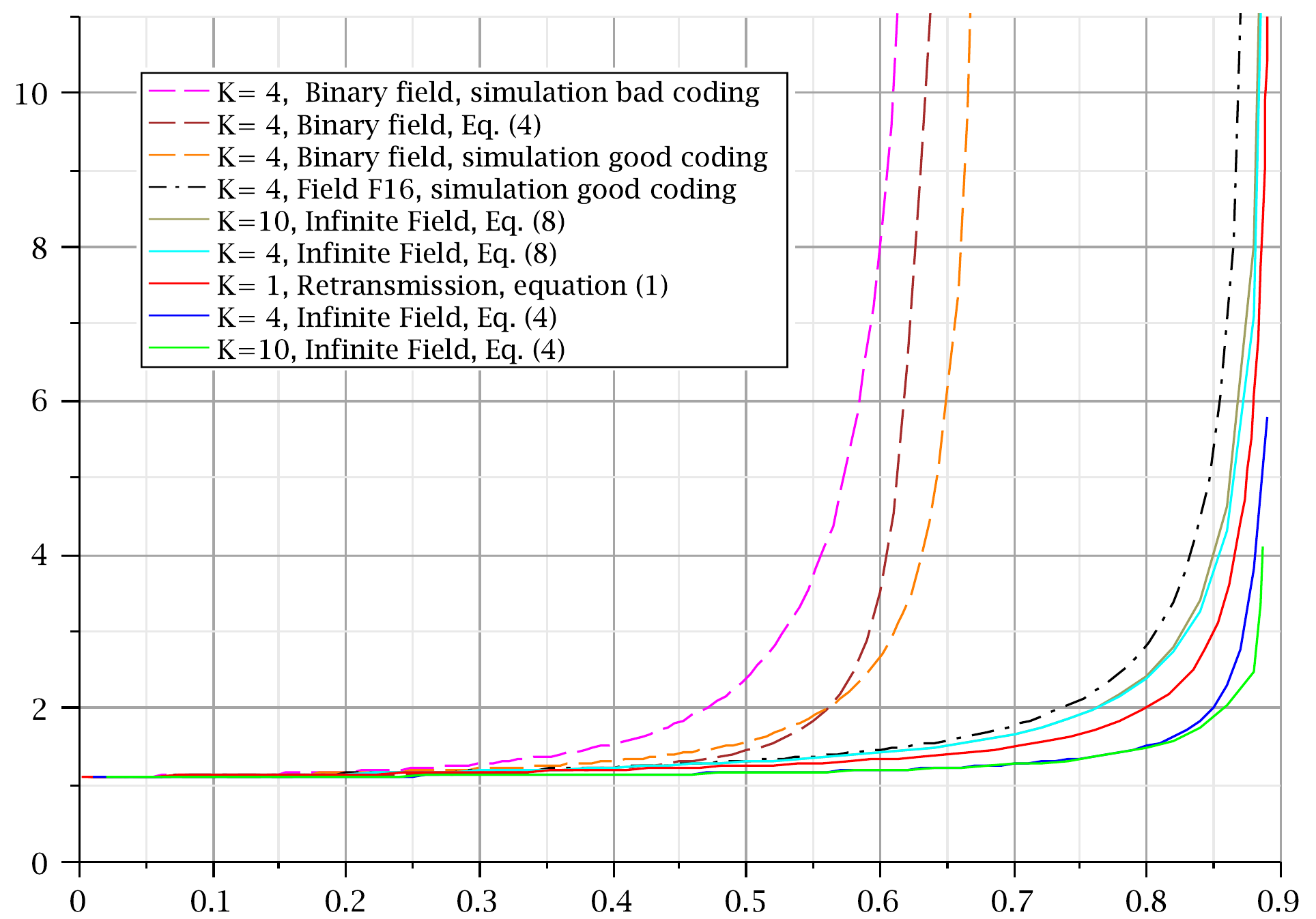}}
%\vspace{-2.0cm}
\caption{Average delay, in time slots, as a function of arrival rate $\lambda$ for $q=0.9$. For explanation: see the caption of Figure~\ref{fig:sim1}.}\label{fig:sim2}
\end{figure}

\subsubsection*{Observations}
\begin{itemize}
  \item 16-ary encoding for $K=4$ according to computer simulations performs very close to RLC($4,\mathbb{F}_\infty$) in terms of delay.
  \item The estimate of RLC($K,\mathbb{F}_\infty$) based on \cite{Shrader-MILCOM07} and repeated here in equation (\ref{eq:DRLCApprox}) decreases with $K$, while in fact the delay increases with $K$.
  \item \cite{Shrader-MILCOM07} underestimates RLC($K,\mathbb{F}_2$), as determined from simulation.
  \item Disallowing all-zero linear combinations makes the difference between  ``bad coding'' and ``good coding'', and can be seen to improve performance significantly for small $K$ and field sizes.
\end{itemize}

\section{Consequences/Implications}

\label{sec:sum}
As noted also in \cite{Shrader-MILCOM07}, random coding for the unicast model studied here will increase the delay as compared with the simple retransmission protocol, and random coding over a finite field and imposing a finite maximum bulk size will even reduce the channel capacity (as defined by the saturation point of the queue at the sender.) However, the loss can be substantially reduced by coding over, say, a field with 16 elements.

The estimate (\ref{eq:DRLCApprox}) (from \cite{Shrader-MILCOM07}) for the delay of the RLC($K,\mathbb{F}_\infty$) suggests that the delay decreases with increasing $K$. In fact equation (\ref{eq:delRLC}) shows that the opposite is true, even though the saturation point remains the same.

% \subsubsection*{Model limitations and future work}
A more realistic queueing model could take into account, for example,
\begin{itemize}
 \item the transmission delay and the round trip delay,
 \item the cost of network resources,
 \item the effect of a lossy feedback channel,
 \item variance of delay.
\end{itemize}
It will be of interest to study a model extended in these directions. For example, the case of a nonzero round trip delay splits the concept of \emph{delay} into two separate issues: \emph{sender delay} and \emph{receiver delay}. For the sender delay, a nonzero round trip delay can be alleviated (at least in terms of queueing effects) by proper allocation of buffer space. However, it is reasonable to believe that a random coding approach may offer benefits with respect to receiver delay, since coding may reduce the expected number of round trip delays per information packet, but coding will also increase the queueing delay at the sender. Thus it requires further investigation to optimize the protocol.
\section*{Acknowledgment}
% optional entry into table of contents (if used)
%\addcontentsline{toc}{section}{Acknowledgment}
The authors would like to thank the Spanish Government
 through project MTM2007-66842-C02-01, Universidad de Valladolid for grant
PM08/33, and the Norwegian Research Council through the OWL
project, for supporting their research.

%\section{Notes...not in final version}
%- Other references:  Jack+Towsley?+other Jack paper, Eric Rozner et al, DeTurck et al, Milstein %et al,

\end{document}